\documentclass[runningheads]{llncs}

\usepackage{xspace}
\usepackage[utf8]{inputenc}
\usepackage{graphicx}
\usepackage{amsmath}
\usepackage{float}
\usepackage{amsfonts}
\usepackage{algorithm}      
\usepackage{algpseudocode}  
\usepackage[breaklinks,bookmarks,bookmarksdepth=1]{hyperref}
\usepackage{enumerate}
\usepackage{paralist}
\usepackage{enumitem}

\newcommand{\numEdges}{m}
\newcommand{\minima}{t}
\newcommand{\Oh}{\mathcal{O}}

\title{Encodings of Range Maximum-Sum Segment Queries and Applications}

\author{Pawe{\l} Gawrychowski\thanks{Currently holding a post-doc position at Warsaw Center of Mathematics and Computer Science.}\inst{1} \and Patrick K. Nicholson\inst{2}}
\institute{Institute of Informatics, University of Warsaw, Poland \and Max-Planck-Institut für Informatik, Saarbr\"ucken, Germany}

\allowdisplaybreaks

\begin{document}

\pagestyle{plain}
\maketitle

\begin{abstract}
Given an array $A$ containing arbitrary (positive and negative)
numbers, we consider the problem of supporting \emph{range maximum-sum
  segment queries} on $A$: i.e., given an arbitrary range $[i,j]$,
return the subrange $[i',j'] \subseteq [i,j]$ such that the sum
$\sum_{k=i'}^{j'} A[k]$ is maximized.\footnote{We use the terms
  segment and subrange interchangeably, but only use segment when
  referring to the name of the problem, for consistency with prior
  work.} Chen and Chao~[Disc.~App.~Math.~2007] presented a data
structure for this problem that occupies $\Theta(n)$ words, can be
constructed in $\Theta(n)$ time, and supports queries in $\Theta(1)$
time.  Our first result is that if only the indices $[i',j']$ are
desired (rather than the maximum sum achieved in that subrange), then
it is possible to reduce the space to $\Theta(n)$ bits, regardless the
numbers stored in $A$, while retaining the same construction and query
time.  Our second result is to improve the trivial space lower bound
for any encoding data structure that supports range maximum-sum
segment queries from $n$ bits to $1.89113n - \Theta(\lg n)$, for
sufficiently large values of $n$.  Finally, we also provide a new
application of this data structure which simplifies a previously known
linear time algorithm for finding $k$-covers: given an array $A$ of
$n$ numbers and a number $k$, find $k$ disjoint subranges $[i_1 ,j_1
],...,[i_k ,j_k]$, such that the total sum of all the numbers in the
subranges is maximized.  As observed by
Cs{\"u}r{\"o}s~[IEEE/ACM~TCBB~2004], $k$-covers can be used to
identify regions in genomes.
\end{abstract}

\section{Introduction}

Many core data structure problems involve supporting \emph{range
  queries} on arrays of numbers: see the surveys of Navarro~\cite{N13}
and Skala~\cite{S13} for numerous examples.  Likely the most heavily
studied range query problem of this kind is that of supporting
\emph{range maximum queries} (resp. \emph{range minimum queries}):
given an array $A$ of $n$ numbers, preprocess the array such that, for
any range $[i,j] \subseteq [1,n]$ we can return the index $k \in
[i,j]$ such that $A[k]$ is maximum (resp. minimum).  These kinds of
queries have a large number of applications in the area of text
indexing~\cite[Section~3.3]{F07}.  Solutions have been proposed to
this problem that achieve $\Theta(n)$ space (in terms of number of
machine words\footnote{In this paper we assume the word-RAM model with
  word size $\Theta(\log n)$ bits.}), and constant query
time~\cite{BF00,D13}.  At first glance, one may think this to be
optimal, since the array $A$ itself requires $n$ words to be stored.
However, if we only desire the index of the maximum element,
rather than the value of the element itself, it turns out that it is
possible to reduce the space~\cite{FH11}.

By a counting argument, it is possible to show that $2n -o(n)$ bits
are necessary to answer range maximum queries on an array of $n$
numbers~\cite[Sec. 1.1.2]{FH11}.  On the other hand, rather
surprisingly, it is possible to achieve this space bound, to within
lower order terms, while still retaining constant query
time~\cite{FH11}.  That is, regardless of the number of bits required
to represent the individual numbers in $A$, we can encode a data
structure in such a way as to support range maximum queries on $A$
using $2n + o(n)$ bits.  The key point is that we need not access $A$
during any part of the query algorithm.  In a more broad sense,
results of this type are part of the area of succinct data
structures~\cite{J89}, in which the aim is to represent a data
structure using space matching the information theoretic lower bound,
to within lower order terms.

In this paper, we consider \emph{range maximum-sum segment
  queries}~\cite{CC07}, where, given a range $[i,j]$, the goal is to
return a subrange $[i',j'] \subseteq [i,j]$ such that $\sum_{k =
  i'}^{j'} A[k]$ is maximized.  Note that this problem only becomes
non-trivial if the array $A$ contains negative numbers.  With a bit of
thought it is not difficult to see that supporting range maximum
queries in an array $A$ can be reduced to supporting range maximum-sum
segment queries on a modified version of $A$ that we get by padding
each element of $A$ with a sufficiently large negative number
(see~\cite{CC07} for the details of the reduction).  However, Chen and
Chao~\cite{CC07} showed that a reduction holds in the other direction
as well: range maximum-sum segment queries can be answered using a
combination of range minimum and maximum queries on several different
arrays, easily constructible from $A$.  Specifically, they show that
these queries can be answered in constant time with a data structure
occupying $\Theta(n)$ words, that can be constructed in linear time.

A natural question one might ask is whether it is possible to improve
the space of their solution to $\Theta(n)$ bits rather than
$\Theta(n)$ words, while still retaining the constant query time.  On
one hand, we were aware of no information theoretic lower bound that
ruled out the possibility of achieving $\Theta(n)$ bits.  On the other
hand, though Chen and Chao reduce the problem to several range maximum
and range minimum queries, they still require comparisons to be made
between various word-sized elements in arrays of size $\Theta(n)$
words in order to make a final determination of the answer to the
query; we review the details of their solution in
Section~\ref{sec:chenchao}.  Therefore, it was not clear by examining
their solution whether further space reduction was possible.

Our first result, presented in Section~\ref{sec:encoding}, is that if
we desire only the indices of the maximum-sum segment $[i',j']$ rather
than the value of the sum itself, then we can achieve constant query
time using a data structure occupying $\Theta(n)$ bits and
constructable in linear time.  There are many technical details, but
the main idea is to sidestep the need for explicitly storing the
numeric arrays required by Chen and Chao to make comparisons by
storing two separate graphs that are judiciously defined so as to be
embeddable in one page.  By a well known theorem of
Jacobson~\cite{J89}, combined with later
improvements~\cite{MR01,GRRR06}, it is known that one-page
graphs---also known as outerplanar graphs---can be stored in a number
of bits that is linear in the total number of vertices and edges, and
be constructed in linear time, while still retaining the ability to
navigate between nodes in constant time.  Navigating these graphs
allows us to implicitly simulate comparisons between certain numeric
array elements, thus avoiding the need to store the arrays themselves.

Our second result, presented in Section~\ref{sec:lowerbound} is to
improve the information theoretic lower bound for this problem.  It is
quite trivial to show that one requires $n$ bits to support range
maximum-sum segment queries by observing that if an array contains
only numbers from the set $\{1,-1\}$ we can recover them via
$\Theta(n)$ queries.  Since there are $2^n$ possible arrays of this
type, the lower bound follows.  In contrast, by an enumeration
argument, we give an improved bound of $1.89113n - \Theta(\lg n)$ bits
when $n$ is sufficiently large.  The main idea is to enumerate a
combinatorial object which we refer to as maximum-sum segment trees,
then bound the number of trees of this type using generating functions.

Our final result, presented in Section~\ref{sec:application}, is a new
application for maximum-sum segment data structures. Given an array
and a number $k$, we want to find a $k$-cover: i.e., $k$ disjoint
subranges with the largest total sum. This problem was first studied
by Csur\"{o}s~\cite{Csuros}, who was motivated by an application in
bioinformatics, and constructed an $\Oh(n\log n)$ time algorithm.
Later, an optimal $\Oh(n)$ time solution was found by Bengtsson and
Chen~\cite{optimal}. We provide an alternative $\Oh(n)$ time solution,
which is an almost immediate consequence of any constant time
maximum-sum segment data structure that can be constructed in linear
time. An advantage of our algorithm is that it can be also used to
preprocess the array just once, and then answer the question for any
$k \in [1,n]$ in $\Oh(k)$ time. We remark that this is related, but
not equivalent, to finding $k$ non-overlapping maximum-sum segments,
and finding $k$ maximum-sum segments. In the latter, one considers all
$\binom{n}{2} + n$ segments ordered non-increasingly according to
their corresponding sums, and wants to select the $k$-th
one~\cite{Liu}. In the former, one repeats the following operation $k$
times: find a maximum-sum segment disjoint from all the previously
chosen segments, and add it to the current set~\cite{Ruzzo}.

\section{Notation and Definitions}

We follow the notation of Chen and Chao~\cite{CC07} with a few minor
changes.  Let $A$ be an array of $n$ numbers.  Let $S(i,j)$ denote the
sum of the values in the range $[i, j]$: i.e., $S(i,j) =
\sum_{k=i}^{j} A[k]$. Let $C$ be an array of length $n$ such that
$C[i]$ stores the cumulative sum $S(1,i)$.  Note that $S(i,j) = C[j] -
C[i-1]$ if $i > 1$.

Given an arbitrary array $B$, a range maximum query
$\textsc{RMaxQ}(B,i,j)$ returns the index of the rightmost maximum
value in the subarray $B[i,j]$; the query $\textsc{RMinQ}(B,i,j)$ is
defined analogously.  A range maximum-sum segment query
$\textsc{RMaxSSQ}(A,i,j)$ returns a subrange $A[i',j']$ such that $i
\le i' \le j' \le j$, and $S(i',j')$ is maximum.  If there is a tie,
then our data structure will return the shortest range with the
largest value of $j'$; i.e., the rightmost
one.\footnote{Alternatively, we can return the leftmost such range by
  symmetry.}  Note that the answer is a range, specified by its
endpoints, rather than the sum of the values in the range.  Also note
that if the range $A[i,j]$ contains only non-positive numbers, we
return an empty range as the solution: we discuss alternatives in
Appendix~\ref{app:emptyrange}, but suggest reading on to
Section~\ref{sec:rmss} first.

\begin{figure}
\centering
\includegraphics[scale=1.3]{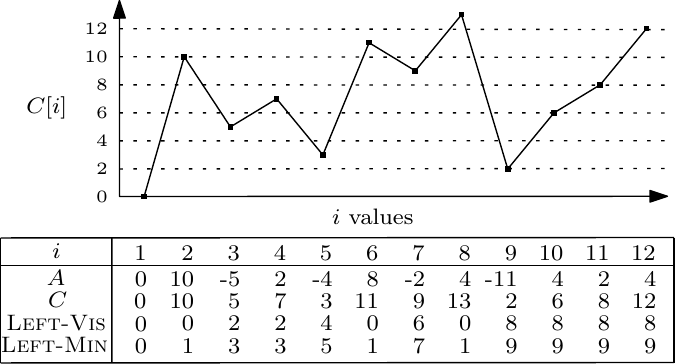}
\caption{\label{fig:definitions}Example array $A$ and the values of
  the various definitions presented in this section that are induced
  by $A$. The list of candidates for this array are: $(1,2)$, $(3,4)$,
  $(1,6)$, $(1,8)$, $(9,10)$, $(9,11)$, $(9,12)$.}
\end{figure}

The \emph{left visible region} $\textsc{Left-Vis}(i)$ of array $C$ at
index $i$ is defined to be the maximum index $1 \le j < i$ such that
$C[j] \ge C[i]$, or $0$ if no such index exists: this corresponds to
the ``left bound'' definition of Chen and Chao.  The \emph{left
  minimum} $\textsc{Left-Min}(i)$ of array $C$ is defined to be
$\textsc{RMinQ}(C,\textsc{Left-Vis}(i) + 1,i)$ for $1 < i \le n$.
This corresponds to (one less than) the ``good partner'' definition of
Chen and Chao.  See Figure~\ref{fig:definitions} for an illustration
of these definitions. The pairs $(\textsc{Left-Min}(i),i)$ where
$\textsc{Left-Min}(i) < i$ are referred to as \emph{candidates}. Thus,
candidates are such that the sum in $A[\textsc{Left-Min}(i)+1..i]$ is
positive. One issue is that the pair $(\textsc{Left-Min}(1),1)$ might
have a positive sum, but not be a candidate.  Without loss of
generality we can ignore this case by assuming $A[1] = 0$, as in
Figure~\ref{fig:definitions}. Define the \emph{candidate score array}
$D$ such that $D[i] = S(\textsc{Left-Min}(i) + 1,i)$ if
$(\textsc{Left-Min}(i),i)$ is a candidate, and $D[i] = 0$ otherwise,
for all $i \in [1,n]$.  Thus, for \emph{non-candidates}, the candidate
score is $0$.  Let $x' = \textsc{RMaxQ}(D,1,n)$ and $\minima' =
\textsc{Left-Min}(x')$. From the definitions it is not too difficult
to see that $\textsc{RMaxSSQ}(A,1,n)$ is $[\minima'+1,x']$ if
$\minima' \neq x'$, and the empty range otherwise.

\section{Preliminary Data Structures}

We make use of the following result for supporting range minimum and
range maximum queries on arrays.

\begin{lemma}[\cite{FH11}]
\label{lem:rmq}
Given an array $B$ of $n$ numbers, we can store a data structure of
size $2n + o(n)$ bits such that $\textsc{RMaxQ}(B,i,j)$ can be
returned for any $1 \le i \le j \le n$ in $\Oh(1)$ time.  Similarly, we can also answer
$\textsc{RMinQ}(B,i,j)$ for $1 \le i \le j \le n$ using the same space
bound.  These data structures can be constructed in linear time so as
to return the rightmost maximum (resp. minimum) in the case of a tie.
\end{lemma}

We also require the following succinct data structure result for
representing one-page graphs.  A one-page (or outerplanar) graph $G$
has the property that it can be represented by a sequence of balanced
parentheses~\cite{J89}.  Equivalently, there exists a labelling $1,
..., n$ of the vertices in $G$ in which there is no pair of edges
$(u_1,u_2)$ and $(u_3,u_4)$ in $G$ such that $1 \le u_1 < u_3 < u_2 <
u_4 \le n$.  That is, if we refer to vertices by their labels, then we
have that the set of ranges represented by edges are either nested or
disjoint: we refer to this as the \emph{nesting property}.  Note that
our definitions of the navigation operations differ (slightly) from
the original source so as to be convenient for our application, and we
explain how to support these operations in
Appendix~\ref{app:one-page}.

\begin{lemma}[Based on~\cite{MR01,GRRR06}]
\label{lem:bp}
Let $G$ be a one-page multigraph with no self-loops: i.e., $G$ has
vertex labels $1, ..., n$, and $m$ edges with the nesting property.
There is a data structure that can represent $G$ using $2(n + m) + o(n
+ m)$ bits, and be constructed in $\Theta(n+m)$ time from the adjacency
list representation of $G$, such that the following operations can be
performed in constant time:

\begin{enumerate}
\item $\textsc{Degree}(G,u)$ returns the degree of vertex $u$.
\item $\textsc{Neighbour}(G,u,i)$ returns the index of the vertex which
  is the endpoint of the $i$-th edge incident to $u$, for $1 \le i \le
  \textsc{Degree}(u)$.  The edges are sorted in non-decreasing order
  of the indices of the endpoints:  $\textsc{Neighbour}(G,u,1) \le
  \textsc{Neighbour}(G,u,2) \le \ldots \le
  \textsc{Neighbour}(G,u,\textsc{Degree}(G,u))$.
\item $\textsc{Order}(G,u,v)$ returns the order of the edge $(u,v)$
  among those incident to $u$: i.e., return an $i$ such that
  $\textsc{Neighbour}(G,u,i) = v$.\footnote{Since $G$ may be a
    multigraph the value of $i$ may be arbitrary among all possible
    values that satisfy the equation.  We note that this is more
    general than we require, as in our application we will not execute
    this type of query on a multigraph, so the answer will always be
    unique.}

\end{enumerate}
In all of the operations above, a vertex is referred to by its label,
which is an integer in the range $[1,n]$.
\end{lemma}

\section{\label{sec:rmss}Supporting Range Maximum-Sum Segment Queries}

In this section we present our solution to the range maximum-sum
segment query problem which occupies linear space in bits.  First we
begin by summarizing the solution of Chen and Chao~\cite{CC07}. Then,
in Section~\ref{sec:q-array} we describe an alternative data structure
that occupies $\Theta(n)$ words of space.  Finally, we reduce the
space of our alternative data structure to linear in bits.

\subsection{\label{sec:chenchao}Answering Queries using $\Theta(n)$ Words}

In the solution of Chen and Chao the following data structures are
stored: the array $C$; a data structure to support
$\textsc{RMinQ}(C,i,j)$ queries; a data structure for supporting
$\textsc{RMaxQ}(D,i,j)$ queries; and, finally, an array $P$ of length
$n$ where $P[i] = \textsc{Left-Min}(i)$.  Thus, the overall space is
linear in words.

The main idea is to examine the candidate $(P[x],x)$ whose right
endpoint achieves the maximum sum in the range $[i,j]$.  If $P[x]+1
\in [i,j]$ then Chen and Chao proved that $[P[x]+1,x]$ is the correct
answer.  However, if $P[x] + 1 \not\in [i,j]$ then they proved that
there are two possible ranges which need to be examined to determine
the answer.  In this case we check the sum for both ranges and return
the range with the larger sum.  The pseudocode for their solution to
answering the query $\textsc{RMaxSSQ}(A,i,j)$ is presented in
Algorithm~\ref{alg:rmaxssq}:

\begin{algorithm}
\caption{Computing $\textsc{RMaxSSQ}(A,i,j)$.}
\label{alg:rmaxssq}
\begin{algorithmic}[1]
\State $x \gets \textsc{RMaxQ}(D,i,j)$
\If{$P[x] = x$}
 \Comment{In this case $x$ is a non-candidate, so $D[x] = 0$}
 \State \Return the empty range
\ElsIf{$P[x] + 1 \ge i$}
 \Comment{In this case $[P[x]+1,x] \subseteq [i,j]$}
 \State \Return $[P[x] + 1,x]$
\Else
 \Comment{In this case $[P[x]+1,x] \not\subseteq [i,j]$}
 \State $y \gets\textsc{RMaxQ}(D,x+1,j)$ 
 \State $\minima\gets\textsc{RMinQ}(C,i-1, x-1)$
  \If{\label{step:sums}$S(\minima + 1,x) > S(P[y] + 1,y)$}
   \State \Return $[\minima + 1,x]$
  \Else
   \State \Return $[P[y] + 1,y]$
  \EndIf
\EndIf
\end{algorithmic}
\end{algorithm}

Items~(\ref{enum:candidate-fits}),~(\ref{enum:max-x})
and~(\ref{enum:two-candidates}) of the following collection of lemmas
by Chen and Chao imply that the query algorithm is correct.  We use
item~(\ref{enum:nested-p}) later.

\begin{lemma}[\cite{CC07}]
\label{lem:cc-prop}
The following properties hold (using the notation from Algorithm~\ref{alg:rmaxssq}):
\begin{enumerate}
\item \label{enum:candidate-fits}If $[P[x]+1,x] \subseteq [i,j]$ then
  $\textsc{RMaxSSQ}(A,i,j)$ is $[P[x]+1,x]$.
\item \label{enum:max-x} The following inequalities hold: $x < P[y]
  \le y$.
\item \label{enum:two-candidates} If $[P[x]+1,x] \not\subseteq [i,j]$
  then $\textsc{RMaxSSQ}(A,i,j)$ is $[P[y]+1,y]$ or $[\minima +1,x]$.
\item \label{enum:nested-p} If $1 \le i < j \le n$ then it cannot be
  the case that $P[i] < P[j] \le i$.  That is, the ranges $[P[i],i]$
  and $[P[j],j]$ have the nesting property for all $1 \le i < j \le
  n$.
\end{enumerate}
\end{lemma}

On line~\ref{step:sums} of Algorithm~\ref{alg:rmaxssq} the sums can be
computed in constant time using the array $C$.  All other steps either
defer to the range maximum or minimum structures, or a constant number
of array accesses.  Thus, the query algorithm takes constant time to
execute.

\subsection{Reducing the Space to $\Theta(n)$ Bits}

Observe that the data structure for answering $\textsc{RMaxQ}$
(resp. $\textsc{RMinQ}$) queries on $D$ (resp. $C$) only requires $2n
+ o(n)$ bits by Lemma~\ref{lem:rmq}; $4n +o(n)$ bits in total for both
structures.  Thus, if we can reduce the cost of the remaining data
structures to $\Theta(n)$ bits, while retaining the correctness of the
query algorithm, then we are done.  There are two issues that must be
overcome in order to reduce the overall space to linear in bits:

\begin{enumerate}
\item The array $P$ occupies $n$ words, so we cannot store it
  explicitly.

\item In the case where $[P[x]+1,x]$ is not contained in $[i,j]$, we
  must compare $S(\minima+1,x)$ and $S(P[y]+1,y)$ without explicitly storing
  the array~$C$.
\end{enumerate}

The first issue turns out to be easy to deal with: we instead encode
the graph $G = ([n],\{(P[x],x)|1 \le x \le n, P[x] < x\}$, which we
call the \emph{candidate graph} using the following lemma:

\begin{lemma}
\label{lem:p-array}
The candidate graph $G$ can be represented using $4n +o(n)$ bits of
space, such that given any $x \in [1,n]$ we can return
$\textsc{Left-Min}(x)$ in $\Oh(1)$ time.
\end{lemma}

\begin{proof}
Item~(\ref{enum:nested-p}) of Lemma~\ref{lem:cc-prop} implies that the
edges of $G$ have the nesting property.  We store $G$ in the data
structure of Lemma~\ref{lem:bp}.  Given $x$, we can retrieve
$\textsc{Left-Min}(x)$ since, if $(\textsc{Left-Min}(x),x)$ is a
candidate, we have that the degree of $x \in V(G)$ is exactly one and
to a vertex with index less than $x$.  If $x$ is a non-candidate, then
the degree is either exactly zero, or any edge from $x$ is to a vertex
with index larger than $x$, and we can return $P[x] = x$. Thus, we can
navigate to~$\textsc{Left-Min}(x)$ in $\Oh(1)$ time.\qed
\end{proof}

From here onward, we can assume that we have access to the array $P$,
which we simulate using Lemma~\ref{lem:p-array}. Unfortunately, the
second issue turns out to be far more problematic.  We overcome this
problem via a two step approach.  In the first step, we define another
array $Q$ which we will use to avoid directly comparing the sums
$S(\minima+1,x)$ and $S(P[y]+1,y)$.  This eliminates the need to store the
array $C$.  We then show how to encode the array $Q$ using $\Theta(n)$
bits.

\subsubsection{\label{sec:q-array}Left Siblings and the $Q$ Array:}

Given candidate $(P[x],x)$, we define the \emph{left sibling}
$\textsc{Left-Sib}((P[x],x))$ to be the largest index $\ell \in
[1,P[x] -1]$, such that there exists an $\ell' \in [\ell+1,P[x]]$ with
$S(\ell+1,\ell') > S(P[x]+1,x)$, if such an index exists.  Moreover,
when discussing $\ell'$ we assume that $\ell'$ is the smallest such
index.  If no such index $\ell$ exists, or if $(P[x],x)$ is a
non-candidate, we say $\textsc{Left-Sib}((P[x],x))$ is undefined.  We
define the array $Q$ such that $Q[x] = \textsc{Left-Sib}((P[x],x))$
for all $x \in [1,n]$; if $Q[x]$ is undefined, then we store the value
$0$ to denote this. We now prove that we can compare $S(\minima+1,x)$
to $S(P[y]+1,y)$ using the $Q$ array.

\begin{lemma}
\label{lem:case-analysis}
If $P[y] = y$ or $Q[y] \ge \minima$ then $\textsc{RMaxSSQ}(A,i,j) =
[\minima+1,x]$.  Otherwise, $\textsc{RMaxSSQ}(A,i,j) = [P[y]+1,y]$
\end{lemma}

\begin{proof}
There are several cases.  For this proof, $\ell'$ is defined relative
to $Q[y]$:
\begin{enumerate}
\item If $P[y] = y$ then $A[y]$ is non-positive.  However, since
  $[P[x]+1,x] \not\subseteq [i,j]$ it is implied that $A[x]$ is
  positive, thus $[\minima+1,x]$ is the correct answer.
\item If $Q[y] > x$ then there exists an $\ell'$ such that
  $S(Q[y]+1,\ell') > S(P[y]+1,y)$ with $\ell' < P[y]$.  Thus,
  $(Q[y],\ell')$ is a better candidate in the range $[x+1,j]$ than
  $(P[y],y)$, so $\textsc{RMaxQ}(D,i+1,j) = \ell'$ and we have a
  contradiction.  Hence, $Q[y] \le x$.
\item If $[Q[y],\ell'] \subseteq [\minima,x]$ then by transitivity we
  have $S(P[y]+1,y) < S(\minima+1,x)$. If $Q[y] \ge \minima$ and
  $\ell' > x$ then it implies that $S(x+1,\ell')$ is positive. This
  implies $D[\ell'] > D[x]$, which contradicts the fact that
  $\textsc{RMaxQ}(D,i,j) = x$. Thus, if $Q[y] \ge \minima$, then
  $[\minima+1,x]$ is the correct answer.
\item If $Q[y] < \minima$ then it is implied that the sum
  $S(\minima+1,x) \le S(P[y]+1,y)$ by the definition of left sibling,
  so $[P[y]+1,y]$ is the rightmost correct answer.  Note that if $Q[y]
  = 0$ and $P[y] \neq y$, then $S(P[y]+1,y)$ is larger than any
  subrange in $[1,P[y]-1]$, so this also holds.  \qed
\end{enumerate}
\end{proof}

Note that in the previous proof, we need not know the value of $\ell'$
in order to make the comparison: only the value $Q[y]$ is required.

\subsubsection{\label{sec:encoding}Encoding the $Q$ Array:}

\begin{figure}[t]\centering
\includegraphics[width=\textwidth]{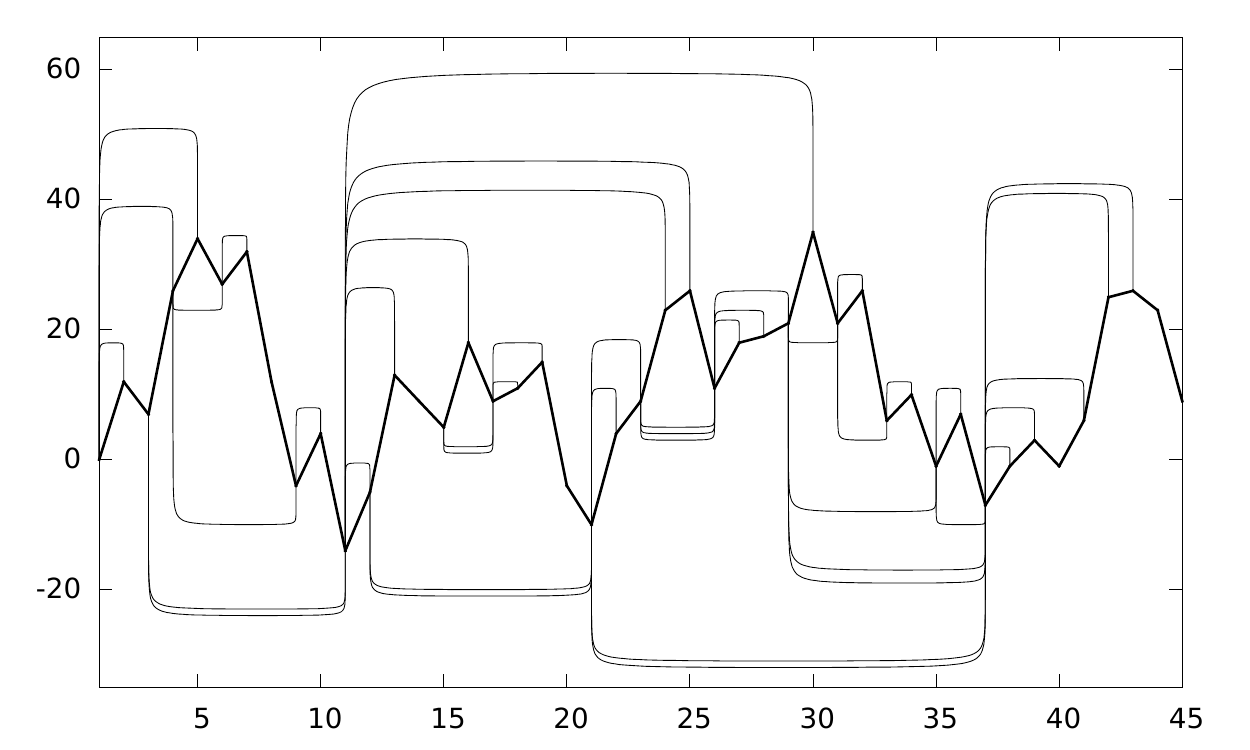}
\caption{\label{fig:left-sib-graph}A graph of the cumulative sums $C$
  (thick middle line) for a randomly generated instance with $n = 45$
  and floating point numbers drawn uniformly from the range $[-20,
    20]$. The $x$-axis is the number $i$, and the $y$-axis is $C[i]$.
  The edges drawn above the line for $C$ represent the candidate graph
  $G$ and the edges below represent the left sibling graph $H$: note
  that $H$ is a multigraph.}
\end{figure}

Unfortunately, the graph defined by the set of edges $(Q[x],x)$ does
not have the nesting property.  Instead, we construct an $n$-vertex
graph $H$ using the pairs $(Q[x],P[x])$ as edges, for each $x \in
[1,n]$ where $\textsc{Left-Sib}(x)$ is defined (i.e., $Q[x] \neq 0$).
We call $H$ the \emph{left sibling graph}.  We give an example
illustrating both the graphs $G$ and $H$ in
Figure~\ref{fig:left-sib-graph}.  Note in the figure that each edge in
$G$ has a corresponding edge in $H$ unless its left sibling is
undefined.  We formalize this intuition in the following lemma:

\begin{lemma} 
\label{lem:cw-order}
Let $(P[x],x)$ be a candidate, and suppose $i =
\textsc{Degree}(H,P[x]) - \textsc{Order}(G,P[x],x) + 1$. If $i > 0$
then it is the case that $\textsc{Left-Sib}((P[x],x)) =
\textsc{Neighbour}(H,P[x],i)$.  Otherwise,
$\textsc{Left-Sib}((P[x],x))$ is undefined.
\end{lemma}

\begin{proof}
Let $\minima = P[x]$ and consider the set of candidates $(\minima,x_1), ...,
(\minima,x_d)$, where $d = \textsc{Degree}(G,\minima)$, and $x_i =
\textsc{Neighbour}(G,\minima,i)$.  Then, we have $S(\minima+1,x_{i-1}) <
S(\minima+1,x_i)$ for all $1 < i \le d$, since $x_{i-1} < x_i$ and $\minima$ is
contained in the left visible region of all $x_i$.  Furthermore,
suppose $\ell_j = \textsc{Left-Sib}((\minima,x_j))$ for all $1 \le j \le
d'$, where $d' = \textsc{Degree}(H,\minima)$.  We show that $\ell_{j-1} \ge
\ell_j $, by assuming the opposite.  If $\ell_{j-1} < \ell_j$ then it
is implied by the definition of left sibling that $S(\minima+1,x_j) <
S(\minima+1,x_{j-1})$, which is a contradiction.  Thus, we have $\ell_{1}
\ge \ell_{2} \ge ... \ge \ell_{d'}$, and the remaining candidates
$(\minima,x_{d'+1}), ..., (\minima,x_{d})$ have undefined left siblings, as there
would be an edge in $H$ corresponding to them otherwise.  The
calculation in the statement of the lemma is equivalent to the
previous statement.  \qed
\end{proof}

Next, we prove the property that can be observed
in Figure~\ref{fig:left-sib-graph}: namely, that we can apply
Lemma~\ref{lem:bp} to $H$.

\begin{lemma}\label{lem:lsg}
The left-sibling graph $H$ can be represented using no more than $4n +
o(n)$ bits of space, such that given any $x \in [1,n]$ we can return
$Q[x]$ in constant time, assuming access to the data structure of
Lemma~\ref{lem:p-array}.
\end{lemma}

\begin{proof}
If $H$ is does not have the nesting property, then there exist
candidates $(\minima_1, x_1)$, $(\minima_2,x_2)$ with $\minima_1 < \minima_2$ such that
$\ell_1 < \ell_2 < \minima_1$, where $\ell_1 = \textsc{Left-Sib}((\minima_1,x_1))$
and $\ell_2 = \textsc{Left-Sib}((\minima_2,x_2))$. Proof by case analysis:

\begin{enumerate}
\item If $\minima_2 = x_1$ then $(\minima_2,x_2)$ is not a valid
  candidate: the candidate with right endpoint $x_2$ would be
  $(\minima_1,x_2)$.
\item If $\minima_2 > x_1$ then we have a contradiction because
  $S(\minima_2+1,x_2) > S(\minima_1+1,x_1)$ (since $\ell_2 <
  \minima_1$), and $\ell_1$ cannot be less than $\ell_2$.
\item If $\minima_2 < x_1$ then we have $x_2 < x_1$ because candidates
  are nested by Lemma~\ref{lem:cc-prop} (Item~\ref{enum:nested-p}).
  However, since $(\minima_2,x_1)$ is not a candidate,
  $C[\minima_2]>C[\minima_1]$, and since $(\minima_1,x_1)$ is a
  candidate we have $C[x_1] > C[x_2]$.  Thus, there must exist some
  $x_3 \in [\minima_1,\minima_2]$ such that $C[x_3] \ge C[x_2]$:
  otherwise $(\minima_1,x_2)$ would be a candidate, and not
  $(\minima_2,x_2)$.  This implies $S(\minima_1+1, x_3) >
  S(\minima_2+1,x_2)$, and therefore we have a contradiction since
  $\ell_2$ cannot be less than $\minima_1$.
\end{enumerate}

\noindent
Thus, we can apply Lemma~\ref{lem:bp} to $H$, achieving the desired
space bound.  Using the calculation in Lemma~\ref{lem:cw-order} we can
return $Q[x]$ in constant time for any candidate $(P[x],x)$.
\qed
\end{proof}

Thus, to simulate the query algorithm of Chen and Chao we need: the
range maximum structure for the array $D$ (Lemma~\ref{lem:rmq}); the
range minimum structure for the array $C$ (Lemma~\ref{lem:rmq}); the
representation of the graph $G$ (Lemma~\ref{lem:bp}); the
representation of the graph $H$ (Lemma~\ref{lem:bp}).  We have the
following theorem:

\begin{theorem}
\label{thm:main}
There is a data structure for supporting range maximum-sum segment
queries, $\textsc{RMaxSSQ}(A,i,j)$ for any $1 \le i \le j \le n$ in
constant time which occupies $12n + o(n)$ bits.
\end{theorem}

\begin{remark}
We note that the constant factor of $12$ in Theorem~\ref{thm:main} is
suboptimal.  As Rajeev Raman~\cite{RRPersonal} has pointed out, the
space for Lemma~\ref{lem:p-array} can be reduced to $2n + o(n)$ bits.
Furthermore, we have also noted that an additional $n$ bits can be
saved by combining the range maximum and minimum encodings for $D$ and
$C$.  However, both of these improvements are quite technical and we
suspect the optimal constant factor is much lower than $9$.  As such,
we leave determination of this optimal constant as future work.
\end{remark}

\section{\label{sec:construction}Constructing the \texorpdfstring{$P$}{P} and \texorpdfstring{$Q$}{Q} Arrays}

Here we show that the arrays $P$ and $Q$ can be constructed
in $\Oh(n)$ time, and therefore the construction time of the data
structure from Theorem~\ref{thm:main} is $\Oh(n)$.  This follows since
we can construct the adjacency list representation of the graphs $G$
and $H$ in linear time from the arrays $P$ and $Q$, respectively.

Constructing the $P$ array is rather straightforward. We only need to
efficiently compute $\textsc{Left-Vis}(i)$ for every $i$. This can be done
using a simple stack-based procedure. Informally, we consider $i=1,2,\ldots,n$
and maintain a stack, where we keep $\textsc{Left-Vis}(i)$,
$\textsc{Left-Vis}(\textsc{Left-Vis}(i))$,
and so on, see Algorithm~\ref{alg:P}. More formally, we define the
\emph{falling staircase}
of a sequence $(a_{1},\ldots,a_{k})$ to be the maximal sequence of
indices $(i_{1},\ldots,i_{s})$ such that
$i_{s}=k$ and $i_{j}=\max\{x :x<i_{j+1}\text{ and
}a_{x}\geq a_{i_{j+1}}\}$ for $j=s-1,\ldots,1$.
Then $\textsc{Left-Vis}(i)$ is the next-to-last element of the falling staircase
of $(C[1],\ldots,C[i])$, and it is easy to see that
Algorithm~\ref{alg:P} maintain such
falling staircase. Having all $\textsc{Left-Vis}(i)$, we can compute
all $\textsc{Left-Min}(i)$
in $\Oh(n)$ total time with range minimum queries (by Lemma~\ref{lem:rmq},
a range minimum structure can be constructed
in $\Oh(n)$ time and answers any query in $\Oh(1)$ time).

\begin{algorithm}
\caption{Computing all $\textsc{Left-Vis}(i)$.}
\label{alg:P}
\begin{algorithmic}
\State $S\gets \emptyset$
\For{$i=1,2,3\ldots,n$}
  \While{$S\neq\emptyset$ and $C[S.\mathrm{top}]< C[i]$}
    \State $S.\mathrm{pop}()$
  \EndWhile
  \State $\textsc{Left-Vis}(i) \gets S.\mathrm{top}$
  \State $S.\mathrm{push}(i)$
\EndFor
\end{algorithmic}
\end{algorithm}

Computing the $Q$ array is more involved. Let $(i_{1},\ldots,i_{s})$
be the falling staircase of $(C[1],\ldots,C[i])$. It partitions
$[1,n]$ into \emph{ranges}
$[1,i_{1}],[i_{1}+1,i_{2}],\ldots,[i_{s-1}+1],i_{s}]$ For each of
these ranges we store the \emph{rising staircase} of the
corresponding range of the $C$ array, where the rising staircase
of $(a_{1},\ldots,a_{k})$ is the falling staircase of
$(-a_{1},\ldots,-a_{k})$. Each of these rising staircases is stored
on a stack, so we keep a stack of stacks, see
Figure~\ref{fig:stacks}. Observe that the rising staircases allow us to extract
$P[i]$ without using the range minimum structure by simply
returning the leftmost element of the last rising staircase.
This is because the sequence we constructed the last rising staircase
for is exactly $C[(\textsc{Left-Vis}(i)+1]),i]$, and the leftmost
element of a rising staircase is the smallest element of the sequence.

\begin{figure}[t]
\includegraphics[width=\textwidth]{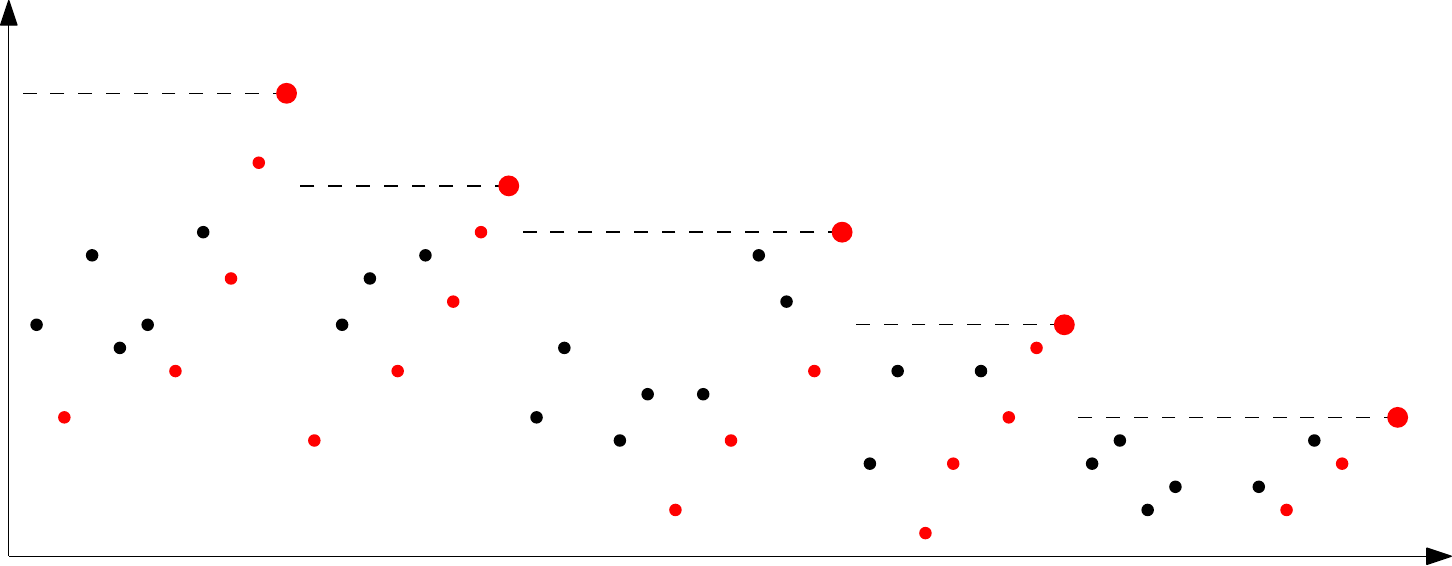}
\caption{A schematic depiction of the stack of rising staircases. Each rising
staircase is in red, and the falling staircase consists of the larger
red elements.}
\label{fig:stacks}
\end{figure}

Before we argue that the rising staircases allow us to compute every $Q[i]$,
we need to show that they can be maintained efficiently as we consider $i=1,\ldots,n$.

\begin{lemma}
\label{lem:rising construction}
The rising staircases constructed for $C[1,i]$ can be updated to be
the rising staircases constructed for $C[1,(i+1)]$ in amortized
$\Oh(1)$ time.
\end{lemma}

\begin{proof}
Consider the rising staircases constructed for $C[1..i]$, and denote
the current ranges by
$[1,i_{1}],[i_{1}+1,i_{2}],\ldots,[i_{s-1}+1,i_{s}]$.  We first find
largest $s'$ such that $C[i_{s'}]>C[i+1]$. Then we pop all
$[i_{s'}+1,i_{s'+1}],\ldots,[i_{s-1}+1,i_{s}]$ from the stack and push
$[i_{s'}+1,i+1]$. The remaining part is to construct the rising
staircase corresponding to the new range $[i_{s'}+1,i+1]$. It can be
constructed using the rising staircases corresponding to the removed
ranges. More precisely, given the rising staircases of
$C[i_{j}+1,i_{j}]$ and $C[i_{j+1}+1..i_{j+2}]$, both stored on stacks,
the rising staircase of $C[i_{j}+1..i_{j+2}]$ can be constructed in
$\Oh(1)$ amortized time.  We call this \emph{merging} the rising
staircases.  It can be implemented by popping the elements from the
latter rising staircase as long as they are smaller than the first
element on the former rising staircase (so the stacks should give us
access to their top elements, but this is easy to add). Every popped
element disappears forever, hence the whole merging procedure
amortizes to $\Oh(1)$ time.  \qed
\end{proof}

Now we are ready show how to compute $Q[i]$ given the rising staircases.
Recall that $Q[i]$ is the largest $\ell < P[i]$ such that there exists
$\ell'\in [\ell,P[i]-1]$ for which $S(\ell+1,\ell')>S(P[i+1]+1,i+1)$.
In other words, we want to find the largest $\ell < P[i]$ such that
$C[\ell']-C[\ell] > C[i]-C[P[i]]$ for some $\ell'\in [\ell,P[i]-1]$.
Because $P[i]$ is the leftmost element of the last rising staircase.
$\ell$ must belong to one of the earlier staircases. Then, because
the rightmost elements on the rising staircases are nonincreasing,
it is enough to consider $\ell$ and $\ell'$ belonging to the same
rising staircase, and furthermore without losing the generality
$\ell'$ is the rightmost element there. So, to summarize, we want
to find the largest $\ell$ belonging to one of the rising staircases
(but not the last one), such that $C[\ell']-C[\ell] > C[i]-C[P[i]]$,
where $\ell'$ is the rightmost element of the same rising staircase.
We will first show how to determine $\ell'$, i.e., the relevant
staircase, and then the rightmost possible $\ell$ there.

We define the \emph{span} of a rising staircase to be the difference
between its leftmost and rightmost element (which, by definition,
is the same as the difference between its smallest and largest element).
We maintain the falling staircase of the sequence of spans of all
rising staircases. In other words, we store the rising staircase
with the largest span, then the rising staircase with the largest span
on its right, and so on. By definition, the span of a rising staircase
is equal to the largest possible value of $C[\ell']-C[\ell]$, where
$\ell$ and $\ell'$ belong to that rising staircase. Therefore, to
determine $\ell'$ we only need to retrieve the rising staircase
corresponding to the next-to-last element of the falling staircase
of the sequence of spans, which can be done in $\Oh(1)$ time,
assuming that we can maintain that falling staircase efficiently.

\begin{lemma}
\label{lem:falling maintenance}
The falling staircase of the sequence of spans of all rising staircases
can be updated in $\Oh(1)$ time after merging the two rightmost
rising staircases.
\end{lemma}

\begin{proof}
Consider the two rightmost rising staircases. After merging, the
rightmost element of the resulting rising staircase is the rightmost
element of the latter rising staircase. The leftmost element of the
resulting rising staircase is either leftmost element of the former
or the latter rising staircase, depending on which one is smaller.
Therefore, the largest element stays the same, and the smallest
element stays the same or decreases, so the span of the new
rising staircase cannot be smaller than the spans of the initial
rising staircases.

Now the falling staircase of the sequence of spans of all rising
staircases can be updated by first popping its elements corresponding
to the two rightmost rising staircases, and then including the
span of the new rising staircase, which might require popping
more elements. Because the new span is at least as large as the
spans of the removed elements, this maintains the falling staircase
correctly in $\Oh(1)$ amortized time.
\qed
\end{proof}

After having determined the appropriate rising staircase, such that
$\ell'$ is the rightmost element there, we want to determine $\ell$.
Denoting the rising staircase by  $(i'_{1},\ldots,i'_{s'})$, where $i'_{s'}=\ell'$,
we need to determine the largest $j$ such that $C[i'_{s'}]-C[i'_{j}] >
C[i]-C[P[i]]$.
This can be done by starting with $j=s'$ and decrementing $j$ as long as
$C[i'_{s'}]-C[i'_{j}] \leq C[i]-C[P[i]]$, i.e., scanning the rising staircase
from right to left. A single scan might require a lot of time, but one can
observe that all scanned elements can be actually removed from the rising
staircase. This is because the next time we scan the same rising staircase
again, the value of $C[i]-C[P[i]]$ will be at least as large as now. When
the rising staircase (or more precisely its prefix) becomes a part of a longer
rising staircase, the scanned elements will be outside of the surviving prefix,
therefore they can be safely removed. This reduces the amortized
complexity of determining a single $\ell$ to $\Oh(1)$, and gives the claimed
total linear time to determine the whole $Q$ array.


\section{\label{sec:lowerbound}Lower Bound}

In this section we prove a lower bound by showing that range maximum
segment sum queries can be used to construct a combinatorial object
which we call \emph{maximum-sum segment trees}, or \emph{MSS-trees}
for short.  By enumerating the total number of distinct MSS-trees, we
get a lower bound on the number of bits required to encode a data
structure that supports range maximum segment sum queries.

\subsection{MSS-trees}

We define MSS-trees as follows.  An MSS-tree for an array $A[1,n]$ is
a rooted ordinal tree, i.e., a rooted tree in which the children are
ordered.  Each node is labelled with a range $[i,j] \subseteq [1,n]$.
For technical reasons, as in the previous sections, we assume $A[1] =
0$.  The intuition is as follows.  Suppose we execute the query
$\textsc{RMaxSSQ}(A,1,n)$ and are given a range $[i+1,j]$.  We define
the \emph{drop} of a query result $[i+1,j]$ to be the range
$[i,j]$---i.e., a range with the left endpoint of the query result
extended by one---as we find it more convenient to discuss drops
rather than query results.  Thus, since $A[1]=0$, all possible drops
span at least two array locations, with the exception of the empty
range, whose drop will be defined to be the empty range. Next, we
consider the partial sums (i.e., the $C$ array), and how we can force
certain drops to occur.  To get a drop of $[i,j]$ we simply fix $C[i]$
to be the minimum, and $C[j]$ to be the maximum.  Then, by fixing
other partials sums we have (roughly) the following flexibility when
setting the values of additional drops in $A$:

\begin{enumerate}

\item \label{enum:left-sit}The drop of $\textsc{RMaxSSQ}(A,1,i-1)$ can
  be completely arbitrary in the range $[1,i-1]$, of length at least
  two, or zero.  Note that it is important that $A[1] = 0$ to make
  this statement true.  Furthermore, we maintain the invariant that
  the values in array locations $C[1]$, $\ldots$, $C[i-1]$ are
  restricted to the range $(C[i],C[j])$, and that the minimum of these
  values occurs to the left of the maximum.

\item \label{enum:right-sit}The drop of $\textsc{RMaxSSQ}(A,j+1,n)$
  can be completely arbitrary in the range $[j+1,n]$, of length at
  least two or zero.  This follows since we know that $A[j+1]$ is a
  non-positive number, since $j+1$ was not contained in $[i,j] =
  \textsc{RMaxSSQ}(A,1,n)$. As in the previous case, we maintain the
  invariant that the values $C[j+1]$, $\ldots$, $[n]$ are restricted to the
  range $(C[i],C[j])$, and that the minimum of these values occurs to
  the left of the maximum.

\item \label{enum:middle-sit}The drop of $\textsc{RMaxSSQ}(A,i,j-1)$
  can be almost completely arbitrary in the range $[i,j-1]$, with
  length at least two. The difference between this case and the
  previous is that the empty range cannot be returned as a drop, nor
  can a drop with left index $i+1$.  This can be seen since
  $\textsc{RMaxSSQ}(A,i,i+1)$ has a drop $[i,i+1]$, as $A[i+1]$ must
  be positive: otherwise, $A[i+1]$ would not be included as the left
  index of the query result for $\textsc{RMaxSSQ}(A,1,n)$ as it does
  not increase the score. Finally, we maintain the invariant that the
  values in array locations $C[i+1]$, $\ldots$, $C[j-1]$ are
  restricted to the range $(C[i],C[j])$, and that the minimum of these
  values occurs to the left of the maximum.

\end{enumerate}

The previous three situations are a bit vague about border cases, and
we will clarify this in our later discussion. Because of these cases,
our MSS-trees will in two flavours: either \emph{general}, which will
describe situations \ref{enum:left-sit} and~\ref{enum:right-sit}, or
\emph{restricted}, which describes situation~\ref{enum:middle-sit}.
General MSS-trees may contain subtrees which are restricted, and
restricted MSS-trees may contain subtrees which are general.  In light
of this mutual definition, we define general MSS-trees first, followed
by restricted MSS-trees.

\subsubsection{General MSS-trees}

Given the array $A$, and a range $[i_0,j_0]$, we construct a general
tree in the following way.  If $i_0=j_0$, then we return a single node
labelled $[i_0,j_0]$.  This is consistent with the fact that we have
enforced the range to begin with a non-positive number for the general
case by setting $A[1]=0$.  Thus, there is only one possible type of
tree when the range has length $1$.  If the range $[i_0,j_0]$ is not
valid (for instance if $i_0 > j_0$), then we return an empty
tree. Otherwise, we execute $\textsc{RMaxSSQ}(A,i_0,j_0)$ and are
given a drop $[i,j] \subseteq [i_0,j_0]$.  We then create a node
labelled with the range $[i,j]$.  The node will have three children
(all of which are possibly empty trees):

\begin{itemize}
\item The left child is a general MSS-tree constructed on the range
  $[i_0,i-1]$.
\item The middle child is a restricted MSS-tree constructed on the
  range $[i,j-1]$.
\item The right child is a general MSS-tree constructed on the range
  $[j+1,j_0]$.
\end{itemize}

\subsubsection{Restricted MSS-tree}

Given a subrange $[i_0,j_0]$, we construct a restricted MSS-tree as
follows.  If $i_0 = j_0$ then we return an empty tree.  If $i_0 =
j_0-1$, then we return a tree labelled with $[i_0,j_0]$.  Otherwise,
if we execute $\textsc{RMaxSSQ}(A,i_0,j_0)$, then we will be given a
drop $[i_0,j] \subseteq [i_0,j_0]$.  This follows from the invariants
since we know that $C[i_0] < C[k]$ for any $k \in [i_0+1,j_0]$.  The
root of the restricted tree is labelled with $[i_0,j]$, and has two
children (again, possibly both empty subtrees):

\begin{enumerate}
\item The left child is the result of recursively constructing a
  restricted MSS-tree on the range $[i_0,j-1]$.
\item The right child is the result of recursively constructing a
  general MSS-tree on the range $[j+1,j_0]$.
\end{enumerate}

\subsection{Examples}

Given any data structure that answers $\textsc{RMaxSSQ}$ queries on
$A$, we can construct the MSS-tree for $A$ by invoking the
construction algorithm for general MSS-trees on the range $[1,n]$.  As
an example, we give a figure showing all possible MSS-trees for $n=3$
and $n=4$ in Figure~\ref{fig:mss-trees}.  Using the invariants
described above, it is not difficult to construct arrays of lengths 3
and 4---in which $A[1] = 0$---such each of the MSS-trees in the figure
can be extracted by the procedure described above.

\begin{figure}
\centering
\includegraphics[width=\textwidth]{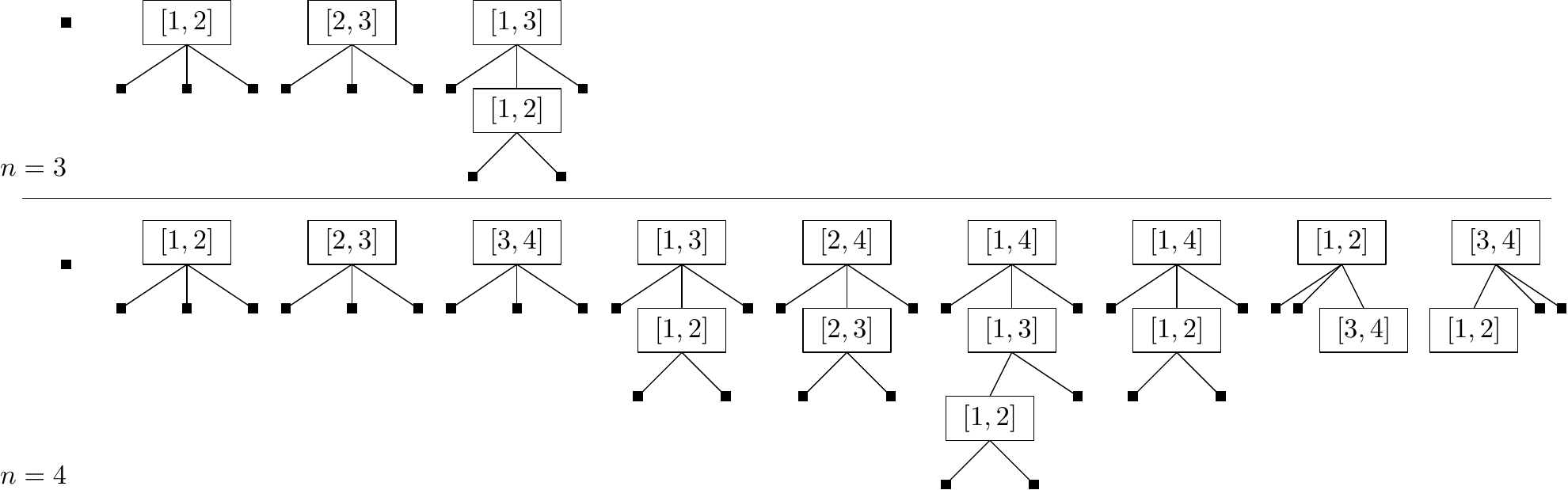}
\caption{\label{fig:mss-trees}Every possible MSS-tree for $n=3$ and
  $n=4$.  Each node is a box labelled with its respective range.  We
  use black boxes to denote empty subtrees.}
\end{figure}

\subsection{Enumeration via Recurrences}

Based on the discussion in the previous section we write the following
recurrences to count the number of MSS-trees for an array of length
$n$.  Let $T(n)$ denote the number of general MSS-trees on an array of
length $n$, and $M(n)$ denote the number of restricted trees on an
array of length $n$.

\begin{align}
T(0) &= 1 &\notag\\
T(n) &= 1+ \sum_{j=2}^{n} \sum_{i=1}^{j-1} T(i-1) \cdot M(j-i) \cdot T(n-j)  & \textrm{ for } n \geq 1 \label{eqn:T} \\
M(0) &= 0 &\notag\\
M(1) &= 1 &\notag\\
M(n) &= \sum_{i=2}^{n} M(i-1) \cdot T(n-i) & \textrm{ for } n \geq 2 \label{eqn:M}
\end{align}
We rewrite (\ref{eqn:T}) as follows:
\begin{align*}
T(n) &= 1 + \sum_{j=2}^{n} \sum_{i=1}^{j-1} T(i-1) \cdot M(j-i) \cdot T(n-j) \\
&= 1 + \sum_{j=2}^{n} T(n-j) \sum_{i=1}^{j-1} T(i-1) \cdot M(j-i)  \\
&= 1 + \sum_{j=2}^{n} T(n-j) \sum_{i=1}^{j-1} M(j-i) \cdot T(i-1)  \\
&= 1 + \sum_{j=2}^{n} T(n-j) \sum_{i=1}^{j-1} M(i) \cdot T(j-i-1)  \\
&= 1 + \sum_{j=2}^{n} T(n-j) \sum_{i=2}^{j} M(i-1) \cdot T(j-i)
\end{align*}
Therefore, by combining with (\ref{eqn:M}) we get a simpler recurrence for $T(n)$:
\begin{align}
T(n) = & 1 + \sum_{j=2}^{n} T(n-j) \cdot M(j) & \notag\\
=& 1 + \sum_{j=0}^{n-2} T(n-2-j) \cdot M(j+2) & \textrm{ for } n \geq 1 \label{eqn:T2}
\end{align}
Because $M(0)=0$, we can then rewrite (\ref{eqn:M}) to get the following equivalent recurrence:
\begin{align}
M(n) &= \sum_{i=0}^{n-1} M(i) \cdot T(n-1-i) & \textrm{ for } n \geq 2 \label{eqn:M2}
\end{align}
Now we define two polynomials $p(x)=\sum_{n=0}^{\infty} M(n)x^{n}$ and $q(x)=\sum_{n=0}^{\infty} T(n)x^{n}$. From (\ref{eqn:T2}) and (\ref{eqn:M2}) we get the following equalities.
\begin{align*}
p(x) =& x + xp(x)q(x) \\
q(x) =& \frac{1}{1-x} + x^{2} q(x) \frac{p(x)-x}{x^{2}}
\end{align*}
After substituting $p(x) = \frac{x}{1-x q(x)}$ we get:
\begin{align*}
q(x) =& \frac{1}{1-x} + x^{2} q(x) \frac{p(x)-x}{x^{2}} \\
=& \frac{1}{1-x} + q(x) \left(\frac{x}{1-xq(x)}-x\right) \\
=& \frac{1}{1-x} + \frac{x^{2}q^{2}(x)}{1-xq(x)} \\
\end{align*}
and then:
\begin{align*}
q(x)(1-x)(1-xq(x)) =& (1-xq(x)) + x^{2}q^{2}(x)(1-x) \\
q(x)-xq^{2}(x)-xq(x)+x^{2}q^{2}(x) = & 1-xq(x)  + x^{2}q^{2}(x)(1-x) \\
q(x)-xq^{2}(x) = & 1 - x^{3}q^{2}(x) \\
q^{2}(x) (x^{3}-x) + q(x) - 1 =& 0
\end{align*}
So finally $q(x) = \frac{1 \pm \sqrt{1-4x(1-x^{2})}}{2x(1-x^{2})}$.
We can eliminate the positive branch through a simple sanity check by
setting $x=0$.  Thus, the generating function for the above sequence
of numbers is:

$$q(x) = \frac{1 - \sqrt{1 - 4x(1 - x^2)}}{2x(1 - x^2)} \enspace .$$

Interestingly, this generating function implies that the number of
valid MSS-trees for $n = 0, 1,2,...$ corresponds to OEIS
A157003\cite{OEIS}.  As for the asymptotics: the first singularity
encountered along the positive real axis for the function $q(x)$ is
located at $x \approx 0.2695944$.  By Pringsheim’s Theorem~\cite[see
  p.226 and Theorems IV.6, IV.7]{FS09} this implies that at least
$\log_2(\frac{1}{0.269594^n \texttt{poly(n)}}) \ge 1.89113n -
\Theta(\lg n)$ bits are required to represent an MSS-tree, provided
$n$ is sufficiently large.  Thus, we have proven the following
theorem:

\begin{theorem}
For an array $A$ of length $n$, any data structure that encodes the
solution to range maximum-sum segment queries must occupy at least
$1.89113n - \Theta(\lg n)$ bits, if $n$ is sufficiently large.
\end{theorem}

\section{\label{sec:application}Application to Computing \texorpdfstring{$k$}{k}-Covers}

Given an array $A$ of $n$ numbers and a number $k$, we want to find
$k$ disjoint segments $[i_{1},j_{1}], \ldots, [i_{k},j_{k}]$, called a
\emph{$k$-cover}, such that the total sum of all numbers inside,
called the score, is maximized.  For $k=1$ (the \textsc{RMaxSSQ}
problem on the entire array) this is a classic exercise, often used to
introduce dynamic programming.  For larger values of $k$, it is easy
to design an $\Oh(nk)$ time dynamic programming algorithm, but an
interesting question is whether we can do better. As shown by
Csurös~\cite{Csuros}, one can achieve $\Oh(n\log n)$ time complexity.
This was later improved to $\Oh(n\alpha(n,n))$~\cite{almost} and
finally to optimal $\Oh(n)$ time~\cite{optimal}.  In this section we
show that, assuming a constant time range maximum-sum segment
structure, which can be constructed in linear time, we can preprocess
the array in time $\Oh(n)$, so that given any $k$, we can compute a
maximum $k$-cover in $\Oh(k)$ time.  This improves the previous linear
time algorithm, which needs $\Oh(n)$ time to compute a maximum
$k$-cover regardless of how small $k$ is, so our algorithm is more
useful when there are multiple different values of $k$ for which we
want to compute a maximum $k$-cover.

We iteratively construct a maximum score $k$-cover for
$k=0,1,2,\ldots,n$. This is possible due to the following property
already observed by Csur\"os.

\begin{lemma}
A maximum score $(k+1)$-cover can be constructed from any maximum
score $k$-cover consisting of intervals
$[i_{1},j_{1}],\ldots,[i_{k},j_{k}]$ in one of the two ways:
\begin{enumerate}
\item adding a new interval $[i_{k+1},j_{k+1}]$ disjoint with all
  $[i_{1},j_{1}],\ldots,[i_{k},j_{k}]$,
\item replacing some $[i_{\ell},j_{\ell}]$ with two intervals
  $[i_{\ell},j'],[i',j_{\ell}]$.
\end{enumerate}
\end{lemma}

As any such transformation results in a valid $(k+1)$-cover, we can
construct a maximum score $(k+1)$-cover by simply choosing the one
increasing the score the most.  In other words, we can iteratively
select the best transformation. Now the question is how to do so
efficiently.

We will first show that the best transformation of each type can be
found in $\Oh(1+k)$ time using the range maximum-sum queries. Assume
that we have both a range maximum-sum and a range minimum-sum query
structure available. Recall that out of all possible transformations
of every type, we want the find the one increasing the score the most.

\begin{enumerate}
\item To add a new interval $[i_{k+1},j_{k+1}]$ disjoint with all
  $[i_{1},j_{1}],\ldots,[i_{k},j_{k}]$ increasing the score the most,
  we guess an index $\ell$ such that the new interval is between
  $[i_{\ell},j_{\ell}]$ and $[i_{\ell+1},j_{\ell+1}]$ (if $\ell=0$ we
  ignore the former and if $\ell=k$ the latter condition).  Then
  $[i_{k+1},j_{k+1}]$ can be found with
  $\textsc{RMaxSSQ}(A,i_{\ell}+1,j_{\ell+1}-1)$.
\item To replace some $[i_{\ell},j_{\ell}]$ with two intervals
  $[i_{\ell},j'],[i',j_{\ell}]$ increasing the score the most, we
  observe that the score increases by $-S(j'+1,i'-1)$, hence we can
  guess $\ell$ and then find $(j'+1,i'-1)$ with
  $\textsc{RMinSSQ}(A,i_{\ell},j_{\ell})$.
\end{enumerate}

For every type, we need $1+k$ calls to one of the structures. If each
call takes constant time, the claimed $\Oh(1+k)$ complexity follows.

We will now show that, because we repeatedly apply the best
transformation, the situation is more structured and the best
transformation of each type can be found faster. To this end we define
a \emph{transformation tree} as follows. Its root corresponds to the
maximum-sum segment $[i,j]$ of the whole $A$, meaning that its weight
is $S(i,j)$, and has up to three children. If $A$ is empty or consists
of only negative numbers, the transformation tree is empty.

\begin{enumerate}
\item The left child is the transformation tree recursively defined
  for $A[1..i-1]$.
\item The middle child is the transformation tree recursively defined
  for $-A[i..j]$, i.e., for a copy of $A[i..j]$ with all the numbers
  multiplied by $-1$.
\item The right child is the transformation tree recursively defined
  for $A[j+1..n]$.
\end{enumerate}

If any of these ranges is empty, we don't create the corresponding
child. Now the transformation tree is closely related to the maximum
score $k$-covers.  

\begin{lemma}
\label{lem:subtree1}
For any $k\geq 1$, a $k$-cover constructed by the iterative method
corresponds to a subtree of the transformation tree containing the
root.
\end{lemma}

\begin{proof}
We apply induction on $k$. For $k=1$, the cover is exactly the
maximum-sum segment, which corresponds to the root of the whole
transformation tree. Now assume that the lemma holds for some $k \ge
1$ and consider how the iterative method proceeds. It is easy to see
that any of the possible transformations, i.e., either adding a new
interval or splitting an existing interval into two, correspond to a
child of a node already in the subtree corresponding to the maximum
$k$-cover by the inductive hypothesis. Hence the lemma holds for $k+1$
and so for all $k\geq 1$.  \qed
\end{proof}

This suggests that a maximum $k$-cover can be found by computing a
maximum weight subtree of the transformation tree containing the root
and consisting of $k$ nodes. Indeed, any such subtree corresponds to a
$k$-cover, and by Lemma~\ref{lem:subtree1} a maximum $k$-cover
corresponds to some subtree. To find a maximum weight subtree
efficiently, we observe the following property of the transformation
tree.

\begin{lemma}
\label{lem:heap}
The transformation tree has the max-heap property, meaning that the
weight of every node is at least as large as the weight of its parent.
\end{lemma}

\begin{proof}
We apply induction on $n$. For the induction step, we need to prove
that the weight of the root is at least as large as the weight of all
of its children. This is immediate in case of the left and the right
child, because the weight of the root is the largest $S(i,j)$ for
$1\leq i\leq j\leq n$, the weight of the left child is the largest
$S(i',j')$ for $1\leq i'\leq j' < i$, and the weight of the right
child is the largest $S(i',j')$ for $j<i'\leq j'\leq n$. The weight of
the middle child, if any, is the largest $-S(i',j')$ for $i\leq i'\leq
j'\leq j$, and to finish the proof we need to argue that any such
$-S(i',j')$ is at most $S(i,j)$. But if $-S(i',j')>S(i,j)$, then
$\frac{S(i,i'-1)+S(j'+1,j)}{2}>S(i,j)$, so either $S(i,i'-1)>S(i,j)$
or $S(j'+1,j)>S(i,j)$. In either case, $S(i,j)$ was not a maximum-sum
segment, a contradiction.  \qed
\end{proof}

Therefore, to find a maximum weight subtree consisting of $k$ nodes,
we can simply choose the $k$ nodes with the largest weight in the
whole tree (we assume that the weights are pairwise distinct, and if
not we break the ties by considering the nodes closer to the root
first). This can be done by first explicitly constructing the
transformation tree, which takes $\Oh(n)$ time assuming a constant
time maximum and minimum range-sum segment structures. Then we can use
the linear time selection algorithm~\cite{median} to find its $k$
nodes with the largest weight. This is enough to solve the problem for
a single value of $k$ in $\Oh(n)$ time.

If we are given multiple values of $k$, we can process each of them in
$\Oh(k)$ time assuming the following linear time and space
preprocessing. For every $i=0,1,2,\ldots,\log n$ we select and store
the $2^{i}$ nodes of the transformation tree with the largest
weight. This takes $\Oh(n+n/2+n/4+...)=\Oh(n)$ total time and space.
Then, given $k$, we find $i$ such that $2^{i}\leq k < 2^{i+1}$ and
again use the linear time selection algorithm to choose the $k$ nodes
with the largest weight out of the stored $2^{i+1}$ nodes.

\bibliographystyle{splncs03}
\bibliography{biblio}

\newpage
\appendix

\section{Alternatives to the Empty Range\label{app:emptyrange}}

Although the data structure we describe returns the empty range in the
case that $[i,j]$ only contains non-positive numbers, it is a simple
modification of our data structure to return the index of the largest
non-positive number instead.  To do this, we keep an additional data
structure that supports range maximum queries on $A$: this occupies
$2n+o(n)$ bits by Lemma~\ref{lem:rmq}.  Whenever our data structure
returns the empty range, we can instead return the result of
$\textsc{RMaxQ}(A,i,j)$ query, which has the desired effect.

\section{\label{app:one-page}Navigation in One-Page Graphs}

We give a description of Munro and Raman's representation, with slight
modification~\cite{MR01}.  Given a one-page graph $G$ there is an
implicit labelling of the vertices from left-to-right along its book
spine.  We represent such a graph using a sequence of balanced
parenthesis sequence $B$.  Each vertex $u$ is represented as a pair
``()'', and the edges incident to the vertex are represented either as
an opening or closing parenthesis sequence $S_u$ that follow this
pair. For the purposes of exposition we orient the edges so that edge
$(u,v)$, where $u < v$ is directed from $u$ to $v$: this is just to
simpify the description, the edges are actually undirected.  Consider
the vertex labelled $u \in [1,n]$.  Each edge directed into $u$ is
represented in the prefix of $S_u$, and each edge directed out of $u$
is represented in the suffix.  Let $v_z = \textsc{Neighbour}(u,z)$.
Then $(v_1,u),(v_2,u),...,(v_x,u)$ are the edges directed into $u$,
and $(u,v_{x+1}), ..., (u,v_{x+y})$ are the edges directed out of $u$,
for some $x+ y = \textsc{Degree}(u)$.  Then the sequence $S_u =
)^x(^y$.  The $i$-th ``)'' from left-to-right represenents the edge
$(v_{x-i +1},u)$, whereas the $i$-th ``(`` from left-to-right
represents the edge $(u,v_{x+i})$.

We construct the data structure of Geary et al.~\cite{GRRR06} on $B$.
We use this particular structure, since has a simple construction
algorithm that takes $\Oh(n+\numEdges)$ deterministic worst-case time
(actually it can be constructed in $o(n+\numEdges)$ time, but we only
need the weaker fact).  The data structure for $B$ occupies
$2(n+\numEdges)+o(n+\numEdges)$ bits: note that $\numEdges$ is not
necessarily $\Oh(n)$ in general since, although the graph $G$ is planar,
it is a multigraph.

We also build rank/select auxiliary structures $W_1$ on the balanced
parenthesis sequence $B$, each ``()'' will represent a $1$ bit and all
other combinations of pairs of parenthesis represent a $0$.  This
takes $o(n+\numEdges)$ bits in total, since we need not store the bit
vector explicitly, just lookup tables~\cite{MR01}: the details here
are rather technical, but mainly involve specialized table lookup.
This allows us, given a label $u$, to jump immediately to the balanced
parenthesis pair ``()'' that represents vertex $u$ using a select
operation on $W_1$.  Similarly, given an arbitrary open/close
parenthesis at position $i$, we can use the rank operation on $W_1$ to
compute the vertex associated with that open/close parenthesis.  Given
this representation, it is easy to compute $\textsc{Degree}(u)$ of
some vertex $u$: we simply return the distance (minus one) between the
ending ``)'' of the pair representing $u$, and the starting ``(''
representing $u+1$.

The representation just described almost allows us to perform the
operation $\textsc{Neighbour}(G,u,i)$.  The issue is that listing the
neighbours in the obvious way returns them in a slightly strange
order: we get the edges directed into $u$ in non-increasing order of
their starting vertex, followed by the edges directed out of $u$ in
non-decreasing order of their starting vertex.  In fact, for our
application this is enough, because in our application all of the
vertices either have no edges directed out, or no edges directed in.
Thus, we can use $\textsc{Degree}(u)$, and rank/select operations on
$W_1$ to return the neighbours in non-decreasing order.  The
$\textsc{Order}(G,u,v)$ operation becomes trivial as well, since, it
is known that if $(u,v)$ is an edge, we can find the pair of
parentheses that represent $(u,v)$ in constant time using the
representation just described.  Once we have either the opening or
closing parenthesis representing $(u,v)$ in $S_u$, we can easily
compute its order using $\textsc{Degree}(u)$ and rank/select
operations on $W_1$.

However, in the more general setting described in the lemma it is not
difficult to perform all these operations as described. In addition to
the previous data structures, we construct auxiliary structures $W_2$
which mark, for each vertex $u$, the point in $S_u$ where the first
opening parenthesis appears.  The $W_2$ structures allow us to perform
rank and select operations on these marked positions.  These also
occupy $o(n+m)$ bits if used in combination with $W_1$.  The details
are, again, rather technical, but the idea is to used specialized
lookup tables, bit masking, and auxiliary rank/select structures.  The
purpose of $W_2$ is to allow us to perform the operations
$\textsc{Order}$ and $\textsc{Neighbour}$, as the closing parenthesis
in the prefix of $S_u$ are stored in non-ascending order of their
endpoints.  It is not difficult to see that using rank and select on
$W_2$, together with a small calculation, allow us to support
$\textsc{Order}$ and $\textsc{Neighbour}$ operations in constant time.

As for construction time, given the adjacency list representation, we
can determine for all vertices $u \in V(G)$ the string $S_u$ in
$\Oh(n+m)$ time in total.  This allows us to write down the balanced
parenthesis sequence in $\Oh(n+m)$ time.


\end{document}